\documentclass[a4paper,UKenglish,cleveref, autoref]{lipics-v2019}
\usepackage[colorinlistoftodos, textsize=footnotesize, backgroundcolor=orange]{todonotes}

\newcommand{\dd}{{\rm d}}
\newcommand{\ex}{{\rm e}}
\newcommand{\cauchy}{C_{\mathcal{D}}}

\newcommand{\rd}{\mathbb{R}^{d}}
\newcommand{\rk}{\mathbb{R}^{k}}

\newcommand{\eps}{\varepsilon}
\newcommand{\norm}[1]{\left \rVert {#1} \right \rVert}

\newcommand{\net}{\mathcal{N}}

\DeclareMathOperator*{\prob}{\mathbf{Pr}}

\DeclareMathOperator*{\EE}{\mathbb{E}}

\title{Near-Neighbor Preserving Dimension Reduction for Doubling Subsets of \texorpdfstring{$\ell_1$}{l1}}

\author{Ioannis Z. Emiris}{Department of Informatics \& Telecommunications, \and
National \& Kapodistrian University of Athens, Greece \and 
ATHENA Research \& Innovation Center, Greece}{emiris@di.uoa.gr}{}{Partially supported by the European Union's H2020 research and innovation programme under grant agreement No. 734242 (LAMBDA).}

\author{Vasilis Margonis}{Department of Informatics \& Telecommunications, \and
National \& Kapodistrian University of Athens, Greece}{basilis.math@gmail.com}{}{}

\author{Ioannis Psarros\footnote{This work was done while the third author was a PhD candidate in the Department of Informatics \& Telecommunications, 
National \& Kapodistrian University of Athens, Greece}}{Institute of Computer Science, University of Bonn, Germany}{ipsarros@uni-bonn.de}{}{Generously supported by the Hausdorff Center for Mathematics.}

\nolinenumbers 

\authorrunning{I.\,Z.\ Emiris, V.\ Margonis and I.\ Psarros}

\Copyright{Ioannis Z.\ Emiris, Vasilis Margonis and Ioannis Psarros}



\supplement{}

\acknowledgements{
IZE is member of team AROMATH, joint between INRIA Sophia-Antipolis and NKUA.  
IP thanks Robert Krauthgamer for useful discussions on the topic.
}

\hideLIPIcs  

\EventEditors{Dimitris Achlioptas and L\'{a}szl\'{o} A. V\'{e}gh}
\EventNoEds{2}
\EventLongTitle{Approximation, Randomization, and Combinatorial Optimization. Algorithms and Techniques (APPROX/RANDOM 2019)}
\EventShortTitle{\mbox{\scriptsize APPROX/RANDOM 2019}}
\EventAcronym{APPROX/RANDOM}
\EventYear{2019}
\EventDate{September 20--22, 2019}
\EventLocation{Massachusetts Institute of Technology, Cambridge, MA, USA}
\EventLogo{}
\SeriesVolume{145}
\ArticleNo{47}

\begin{document}

\maketitle

\begin{abstract}
Randomized dimensionality reduction has been recognized as one of the fundamental techniques in handling high-dimensional data. Starting with the celebrated Johnson-Lindenstrauss Lemma, such reductions have been studied in depth for the Euclidean $(\ell_2)$ metric, but much less for the Manhattan $(\ell_1)$ metric. 
Our primary motivation is the approximate nearest neighbor problem in $\ell_1$. We exploit its reduction to the decision-with-witness version, called approximate \textit{near} neighbor, which incurs a roughly logarithmic overhead.
In 2007, Indyk and Naor, in the context of approximate nearest neighbors, introduced the notion of nearest neighbor-preserving embeddings. These are randomized embeddings between two metric spaces with guaranteed bounded distortion only for the distances between a query point and a point set.
Such embeddings are known to exist for both $\ell_2$ and $\ell_1$ metrics, as well as for doubling subsets of $\ell_2$. 
The case that remained open were doubling subsets of $\ell_1$. 
In this paper, we propose a dimension reduction by means of a \textit{near} neighbor-preserving embedding for doubling subsets of $\ell_1$. 
Our approach is to represent the pointset with a carefully chosen covering set, then randomly project the latter. We study two types of covering sets: $c$-approximate $r$-nets and randomly shifted grids, and we discuss the tradeoff between them in terms of preprocessing time and target dimension.
We employ Cauchy variables: certain concentration bounds derived {should be of independent interest}.
\end{abstract}

\section{Introduction}

Proximity search is a fundamental computational problem with several applications in Computer Science and beyond. Proximity problems in metric spaces of low dimension have been typically handled by methods which discretize the space and therefore are affected by the curse of dimensionality, making them unfit for high-dimensional spaces. In the past two decades, the increasing need for analyzing high-dimensional data led researchers to devise randomized and approximation algorithms with polynomial dependence on the dimension.

A fundamental proximity problem is Approximate Nearest Neighbor search. By known reductions~\cite{HIM12}, one can (up to polylogarithmic factors) focus on the decision version with witness, namely the $(c,R)$-Approximate Near Neighbor problem:

\begin{definition}[Approximate Near Neighbor]
Let $(X,d_X)$ be a metric space. Given $P \subseteq X$ and reals $R>0$, $c \geq 1$, build a data structure $\mathcal{S}$ that, given a query point $q \in X$, performs as follows:
\begin{itemize}
\itemsep0em
    \item If the nearest neighbor of $q$ lies within distance at most $R$, then $\mathcal{S}$ is allowed to report any point $p^* \in P$ such that $d_X(q,p^*) \leq cR$.
    \item If all points lie at distance more than $cR$ from $q$, then  $\mathcal{S}$ should return $\bot$.
\end{itemize}
In general, $\mathcal{S}$ returns either a point at distance $\leq cR$ or $\bot$, even when none of the above two cases occurs.
\end{definition}

From now on, we assume $R=1$ because we can re-scale the data set, and we refer to this problem as $c$-ANN, or simply ANN. We focus on subsets of $\ell_1^d$: the input dataset consists of $n$ vectors in $\rd$ and the distance function is the standard $\ell_1$ norm $\| \cdot \|_1$. Note that all logarithms are base $2$.

\subparagraph{Previous work.}
Some highlights in the study of data structures for high-dimensional normed spaces are the various variants, proofs,  and applications of the Johnson Lindenstrauss Lemma (e.g.~\cite{Ach03,AC09,AEP18}), 
sketches based on $p$-stable distributions \cite{I06}, and Locality Sensitive Hashing (e.g.~\cite{IM98,AI08,ALRW17}). In the core of most high-dimensional solutions lies the fact that for certain metric spaces e.g. $\ell_p, p \in [1,2]$, the distance can be efficiently sketched. Spaces which are considered to be harder in this context, such as $\ell_{\infty}$, can also be treated \cite{I01}, and are very interesting since they can be used as host spaces for various norms \cite{ANNRW17}.

Significant amount of work has been undertaken for pointsets of low doubling dimension, since it is today one of the primary paradigms for capturing input structure (formal definitions in the next section). 
For any finite metric space $X$ of doubling dimension $\mathrm{dim}(X)$, there exists a data structure \cite{HPM05,CG06} with
expected preprocessing time $O(2^{\mathrm{dim}(X)} n \log{n})$,
space usage $O(2^{\mathrm{dim}(X)} n)$ (or even $O(n)$) and query time
$O(2^{\mathrm{dim}(X)}\log{n}+\eps^{-O(\mathrm{dim}(X)})$. 

In \cite{IN07}, they introduced the notion of nearest-neighbor preserving embeddings, and it was proven that in this context one can achieve dimension reduction for doubling subsets of $\ell_2$, with the target dimension depending only on the dataset's doubling dimension. Even before, Indyk~\cite{I06} had introduced a randomized embedding for dimension reduction in $\ell_1$, which is suitable for proximity search purposes, and it achieves target dimension polylogarithmic in the size of the pointset.  
Naturally, such approaches can be easily combined with any known data structure  {to be used in the projection space}.  Randomized embeddings have been recently used in the ANN context \cite{BG18}, for doubling subsets of $\ell_p$, $2<p<\infty$.

It is known that dimension reduction in $\ell_1$ cannot be achieved in the same generality as in $\ell_2$, even assuming that the pointset is of low doubling dimension \cite{LMN05}: {there are arbitrarily large $n$-point subsets $P \subseteq \ell_1$ which are doubling with constant $6$, such that every embedding with distortion $D$ of $P$ into $\ell_1^k$ requires dimension $n^{\Omega(1/D^2)}$}. Aiming for more restrictive guarantees, e.g.\ preserving distances within some pre-defined range, is a relevant workaround. Then, dimension reduction techniques for doubling subsets of $\ell_p$, $p\in [1,2]$, exist \cite{BG16}, but they rely on partition algorithms which require the whole pointset to be known in advance. Hence, applicability of such techniques is quite limited and, specifically, it is not clear whether they can be used in an online setting where query points are not known beforehand. 

\subparagraph{Contribution.}
In this paper, we establish two non-linear \textit{near} neighbor-preserving embeddings for doubling subsets of $\ell_1^d$. We use a definition which is essentially a modified version of the nearest neighbor preserving embedding of \cite{IN07}: the guarantees which are required are weaker since we consider the decision version of the problem,  therefore the embedding depends on some range parameter $R>0$. 

\begin{definition}[Near-neighbor preserving embedding]  \label{Dnnpres}
Let $(Y,d_Y)$, $(Z,d_Z)$ be metric spaces and $X \subseteq Y$. A distribution over mappings $f:Y \rightarrow Z$ is a \emph{near-neighbor preserving embedding} with range $R>0$, distortion $D \geq 1$ and probability of correctness $\mathcal{P} \in [0,1]$ if for every $\alpha \geq D$ and any $q \in Y$, if $x \in X$ is such that $d_Y(x,q) \leq R$, then 
with probability at least $\mathcal{P}$, 
\begin{itemize}
    \item $d_Z(f(x),f(q))\leq D  \cdot R$, 
    \item $\forall p \in X:~ d_{Y}(p,q) > D \cdot \alpha \cdot R \implies d_{Z}(f(p),f(q)) > \alpha \cdot R$.
\end{itemize}
\end{definition}

Considering a pointset $P \subset \ell_1^d$ of cardinality $n$, our results concern $\ell_1^k$ as the target space, where $k$ depends on the doubling dimension of $P$. We assume that $R=1$, since we can rescale the dataset. More specifically:

\begin{enumerate} \itemsep1em
    \item In Theorem \ref{theorem:main}, we prove that for every $\eps \in (0,1/2)$ and $c \geq 1$, there is a randomized mapping $h: \ell_1^d \rightarrow \ell_1^k$ that can be computed in time $\tilde{O}(dn^{1+1/ \Omega(c)})$ and is \textit{near} neighbor-preserving for $P$ with distortion $1{+}6\eps$ and probability of correctness $\Omega(\eps)$, where
    \[ k = \left( \log{\lambda_P} \cdot \log(c / \eps) \right)^{\Theta(1/\eps)} / \zeta(\eps), \]
    for a function $\zeta(\eps)>0$ depending only on $\eps$.
    Although the mapping $h$ depends on the pointset, the parameter $c$ is user-defined and therefore provides a trade-off between preprocessing time and target dimension.
    
    \item In Theorem \ref{theorem:second}, we show that for every $\eps \in (0,1/2)$, there is a randomized mapping $h': \ell_1^d \rightarrow \ell_1^k$ that can be computed in time $O(dkn)$ and is \textit{near} neighbor-preserving for $P$ with distortion $1{+}6 \eps$ and probability of correctness $\Omega(\eps)$, where
    \[ k = \left( \log{\lambda_P} \cdot \log(d / \eps) \right)^{\Theta(1/\eps)} / \zeta(\eps),  \]
    for a function $\zeta(\eps)>0$ depending only on $\eps$. 
    In this case, the function $h'$ is oblivious to $P$ and well-defined over the whole space, but the target dimension depends on $d$.
\end{enumerate}

On the low-preprocessing-time extreme, one can embed the dataset in near-linear time, but the target dimension is polynomial in $\log \log n$. This is to be juxtaposed to the analogous result by Indyk \cite{I06}, which provides with target dimension  polynomial in $\log n$, without any assumption on the doubling dimension of the dataset. 
On the other hand, one can obtain a preprocessing time of $dn^{1+\delta}$ for any constant $\delta>0$, and target dimension which depends solely on the doubling dimension. 

\subparagraph{Techniques.} Both embeddings consist of two basic components. First, we represent the pointset $P$ with an $\eps$-covering set, and then we apply a random linear projection \`{a} la Indyk \cite{I06} to that set, using Cauchy variables.

The role of the covering set is to exploit the doubling dimension of $P$. In the analogous result for $\ell_2$ \cite{IN07}, no representative sets were used; the mapping was just a random linear projection of $P$. In the case of $\ell_1$ however, a similar analysis of a linear projection with Cauchy variables without these representative sets seems to be impossible, since the Cauchy distribution is heavy tailed.

In Theorem \ref{theorem:main}, we consider $c$-approximate $r$-nets as a covering set. Inspired by the algorithm of \cite{EHS15} for $\ell_2$, we design an algorithm that computes a $c$-approximate $r$-net in $\ell_1$ in subquadratic --but superlinear-- time. On the other hand, Theorem \ref{theorem:second} relies on randomly shifted grids, which can be computed in linear time, but are inferior to nets in terms of capturing the doubling dimension of the pointset. 

To bound the distortion incurred by the randomized projection, we exploit the 1-stability property of the Cauchy distribution. To this end, we prove a concentration bound for sums of independent Cauchy variables that should be of interest beyond the scope of this paper. To overcome the technical difficulties associated with the heavy tails of the Cauchy distribution, we study sums of \textit{square roots} of Cauchy variables, where in \cite{I06}, Indyk considers sums of \textit{truncated} Cauchy variables instead. Although our concentration bound is rather weak, it is sufficient for our purposes and its analysis is much simpler compared to Indyk's.

\subparagraph{Algorithmic implications.} 
Our results show that efficient dimension reduction for doubling subsets of $\ell_1$ is possible, in the context of ANN. In particular, these results imply efficient sketches, meaning that one can solve ANN with minimal storage per point. Dimension reduction also serves as a problem reduction from a high-dimensional hard instance to a low-dimensional easy instance. Since the algorithms presented in this paper are quite simple, they should also be of practical interest: they easily extend the scope of any implementation which has been optimized to solve the problem in low dimension, so that it may handle high-dimensional data.

Our embedding can be combined with the bucketing method of \cite{HIM12} for the $(1{+} \eps)$-ANN problem in $\ell_1^d$. For instance, setting $c = \log n$ in Theorem \ref{theorem:main}, yields preprocessing time $ dn^{1+ o(1)}$, space $n^{1+ o(1)}$ and query time $O(d){\cdot}(\log{\lambda_P} \cdot \log \log n)^{O(1/ \eps)}$ assuming that the  doubling dimension is a fixed constant. This improves upon existing results: the query time of \cite{KL04} depends on the aspect ratio of the dataset, while 
the data structures of \cite{HPM05,CG06} support queries with time complexity which depends exponentially on the doubling dimension. 
However, it is worth noting that one could potentially improve the results of \cite{KL04,HPM05,CG06} in the special case of $\ell_1$, by employing ANN data structures with fast 
query time, in order to accelerate the traversal of the net-tree. Hence, while our result gives a simple framework for exploiting the intrinsic dimension of doubling subsets of $\ell_1$, it is unlikely that it shall improve upon simple variants of previous results in terms of complexity bounds. 

\subparagraph{Organization.}
The next section introduces basic concepts and some relevant existing results. Section~\ref{Sconc-bound} establishes a concentration bound on sums of independent Cauchy variables.
Section~\ref{Sdimreduc}, achieves dimensionality reduction by means of representing the pointset by a carefully chosen net, while Section~\ref{Sgrids} employs randomly shifted grids for the same task.
We conclude with discussion of results and potential improvements.
\section{Preliminaries}
In this section, we define basic notions about doubling metrics and present useful previous results. 

\begin{definition}
Consider any metric space $(X,\dd_{X})$ and let $B(p,r)=\{x\in X \mid \dd_{X}(x,p)\leq r\}$. The \emph{doubling constant} of $X$, denoted $\lambda_X$, is the smallest integer $\lambda_X$ such that for any $p \in X$ and $r > 0$, the ball $B(p, r)$ can be covered by $\lambda_X$ balls of radius $r/2$ centered at points in $X$. 
\end{definition}

The \emph{doubling dimension} of $(X,\dd_{X})$ is defined as $\log{\lambda_X}$. Nets play an important role in the study of embeddings, as well as in designing efficient data structures for doubling metrics. 

\begin{definition}
For $c \geq 1$, $r >0$ and metric space $(V,d_V)$, a $c$-approximate $r$-net of $V$ is a subset $\mathcal{N} \subseteq V$ such that no two points of $\mathcal{N}$ are within distance $r$ of each other, and every point of $V$ lies within distance at most $c{\cdot}r$ from some point of $\mathcal{N}$.
\end{definition}

\begin{theorem}
\label{theorem:apprnets}
Let $P\subset \ell_1^d$ such that $|P|=n$. Then, for any $c>0$, $r>0$, one can compute a $c$-approximate $r$-net of $P$ in time $\tilde{O}(dn^{1+1/c'})$, where $c'=\Omega(c)$. The result is correct with high probability. The algorithm also returns the assignment of each point of $P$ to the point of the net which covers it.
\end{theorem}

\begin{proof}
We employ some basic ideas from \cite{HIM12}. An analogous result for $\ell_2$ is stated in \cite{EHS15}. First, we assume $r=1$, since we are able to re-scale the point set. Now, we consider a randomly shifted grid with side-length $2$. 
The probability that two points $p,q\in P$ fall into the same grid cell, is at least $1-\|p-q\|_1/2$. For each non-empty grid cell we snap points to a grid: each coordinate is rounded to the nearest multiple of $\delta=1/10dc$.
Then, coordinates are multiplied by $1/\delta$ and each point $x=(x_1,\ldots,x_d) \in [2\delta]^d$ is mapped to $\{0,1\}^{2d /\delta}$ by a function $G$ as follows:
$G(x)=(g(x_1),\ldots,g(x_d))$, 
where $g(z)$ is a binary string of $z$ ones followed by $2/\delta -z$ zeros. For any two points $p,q$ in the same grid cell, let $f(p)$,$f(q)$ be the two binary strings obtained by the above mapping. Notice that,
\begin{equation*}
\|f(p)-f(q)\|_1\in (2/ \delta ) \cdot \|p-q\|_1  \pm 1. 
\end{equation*}
Hence,
\begin{align*}
    \|p-q\|_1 \leq 1 &\implies \|f(p)-f(q)\|_1 \leq (2/ \delta )+1 ,\\[2pt]
    \|p-q\|_1 \geq c &\implies \|f(p)-f(q)\|_1 \geq (2/ \delta )\cdot c-1 .
\end{align*}

Now, we employ the LSH family of \cite{HIM12}, for the Hamming space. After standard concatenation, we can assume that the family is $ (\rho,c'\rho,n^{-1/c'},n^{-1})$-sensitive, where $\rho=\left({2}/{\delta}\right)+1$ and $c'= \Omega(c)$. Let $\alpha=n^{-1/c'}$ and $\beta=n^{-1}$. 

Notice that for the above two-level hashing table we obtain the following guarantees. 
Any two points $p,q\in P$, such that $\|p-q\|_1\leq 1$, fall into the same bucket with probability $\geq \alpha/2$. Any two points $p,q\in P$, such that $\|p-q\|_1\geq c$, fall into the same bucket with probability $\leq \beta$. 

Finally, we independently build $k=\Theta(n^{1/c'} \log n)$ hashtables as above, where the random hash function is defined as a concatenation of the function which maps points to their grid cell id and one LSH function. We pick an arbitrary ordering $p_1,\ldots,p_n\in P$. We follow a greedy strategy in order to compute the approximate net. We start with point $p_1$, and we add it to the net. We mark all (unmarked) points which fall at the same bucket with $p_1$, in one of the $k$ hashtables, and are at distance $\leq c r$. Then, we proceed with point $p_2$. If $p_2$ is unmarked, then we repeat the above. Otherwise, we proceed with $p_3$. The above iteration stops when all points have been marked. Throughout the procedure, we are able to store one pointer for each point, indicating the center which covered it. 

\subparagraph{Correctness.} The probability that a good pair $p,q$ does not fall into the same bucket for any of the $k$ hashtables is $\leq \left(1-\alpha/2\right)^k \leq n^{-10}$. Hence, with high probability, the packing property holds, and the covering property holds because the above algorithm stops when all points are marked. 

\subparagraph{Running time.} The time to build the $k$ hashtables is $k\cdot n =\tilde{O}(n^{1+1/c'})$. Then, at most $n$ queries are performed: for each query, we investigate $k$ buckets and the expected number of false positives is $\leq k\cdot n^2 \cdot \beta=\tilde{O}(n^{1+1/c'})$. Hence, if we stop after having seen a sufficient amount of false positives, we obtain time complexity $\tilde{O}(n^{1+1/c'})$ and the covering property holds with constant probability. We can repeat the above procedure $O(\log n)$ times to obtain high probability of success. 
\end{proof}

The main result in the context of randomized embeddings for dimension reduction in $\ell_1^d$ is the following theorem, which exploits the $1$-stability property of Cauchy random variables and provides with an asymmetric guarantee: The probability of non-contraction is high, but the probability of non-expansion is constant. Nevertheless, this asymmetric property is sufficient for proximity search. 

\begin{theorem}[Thm 5, \cite{I06}]
\label{theorem:indyk}
For any $\eps \leq 1/2$, $\delta > 0$, $\eps > \gamma >0$ there is a probability space over linear mappings $f: \ell_1^d \rightarrow \ell_1^k$, where $k = (\ln{(1/\delta)})^{1/(\eps - \gamma)} / \zeta(\gamma)$, for a function $\zeta(\gamma)>0$ depending only on $\gamma$, such that for any pair of points $p,q \in \ell_1^d$:
\begin{align*} 
    \prob \big[\norm{f(p) - f(q)}_1 \leq (1 - \eps) \norm{p-q}_1 \big] &\leq \delta, \\
    \prob \big[\norm{f(p) - f(q)}_1 \geq (1 + \eps) \norm{p-q}_1 \big] &\leq \frac{1+\gamma}{1+ \varepsilon}.
\end{align*}
\end{theorem}

Note that the embedding is defined as $f(u) = Au/T$, where $A$ is a $k{\times}d$ matrix with each element being an i.i.d.\ Cauchy random variable. In addition, $T$ is a scaling factor defined as the expectation of a sum of truncated Cauchy variables, such that $T = \Theta(k \log{(k/ \eps)})$ (see Lemma 5 in \cite{I06}).

One key observation here is that given a pointset $P$ in a space of bounded aspect ratio $\Phi$, one can directly employ Theorem \ref{theorem:indyk}. The number of points can be upper bounded by a function of $\lambda_P$ and $\Phi$, and hence the new dimension, $k$, depends only on these parameters. This paper proves better bounds than the ones of Theorem \ref{theorem:indyk} for doubling subsets of $\ell_1^d$, without any assumption on the aspect ratio.
\section{Concentration bounds for Cauchy variables}\label{Sconc-bound}
In this section, we prove some basic properties of the Cauchy distribution, which serves as our main embedding tool.  

Let $\cauchy$ denote the Cauchy distribution with density $c(x) = (1/ \pi)/(1+x^2)$. One key property of the Cauchy distribution is the so-called $1$-stability property: 
Let $v=(v_1,\ldots,v_k)\in \rk$ and $X_1,\ldots,X_k$ be i.i.d.\ random variables following $\cauchy$, then $\sum_{j=1}^k X_iv_i$ is distributed as $X {\cdot} \|v\|_1$, where $X \sim \cauchy$. 

The Cauchy distribution has undefined mean. However, for $0 < q < 1$, the mean of the $q$-th power of a Cauchy random variable can be defined. More specifically, for some $X \sim \cauchy$ we have
\begin{equation*} \label{exphalf}
    \EE \left[ |X|^{1/2} \right] =  \frac{2}{\pi} \int_{0}^{\infty} \frac{\sqrt{x}}{1+x^2} ~\mathrm{d} x = \frac{2}{\pi} \frac{\pi}{\sqrt{2}} = \sqrt{2}.
\end{equation*}

The following lemma provides a bound for the moment-generating function of $|X|^{1/2}$.
\begin{lemma} \label{lemma:mgf}
Let $X \sim \cauchy$. Then for any $\beta > 1$:
\begin{equation*}
    \EE \left[ \exp{(-\beta |X|^{1/2})} \right] \leq  \frac{2}{\beta}.
\end{equation*}
\end{lemma}

\begin{proof}
For any constant $\beta$,
\begin{equation*}
    \int_{0}^{1} \ex^{-\beta x^{1/2}} ~\mathrm{d} x = \frac{2}{\beta^2}\left(1 - \frac{\beta + 1}{\ex ^{\beta}} \right) .
\end{equation*}
Then, for any $\beta > 1$,
\begin{align*}
    \EE \left[ \exp{(-\beta |X|^{1/2})} \right] &= \int_{-\infty}^{\infty} \ex^{-\beta |x|^{1/2}} \cdot c(x) ~\mathrm{d} x
    =  \frac{2}{\pi} \int_{0}^{\infty} \ex^{-\beta x^{1/2}} \cdot \frac{1}{1+x^2} ~\mathrm{d} x \\
    &=  \frac{2}{\pi} \int_{0}^{1} \ex^{-\beta x^{1/2}} \cdot \frac{1}{1+x^2} ~\mathrm{d} x + \frac{2}{\pi} \int_{1}^{\infty} \ex^{-\beta x^{1/2}} \cdot \frac{1}{1+x^2} ~\mathrm{d} x  \\
    &\leq \frac{2}{\pi} \int_{0}^{1} \ex^{-\beta x^{1/2}} ~\mathrm{d} x + \frac{2}{\pi} \int_{1}^{\infty} \ex^{-\beta} \cdot \frac{1}{1+x^2} ~\mathrm{d} x \\
    &= \frac{2}{\pi} \cdot \frac{2}{\beta^2}\left(1 - \frac{\beta + 1}{\ex ^{\beta}} \right) +  \frac{1}{2\ex^{\beta}} \\
    &\leq  \frac{4}{\pi \beta^2}+\frac{1}{2\ex^{\beta}} \\
    &\leq \frac{2}{\beta} . \qedhere
\end{align*}
\end{proof}

Let $S := \sum_{j=1}^k |X_j|$ where each $X_j$ is an i.i.d. Cauchy variable. To prove concentration bounds for $S$, we study the sum $\tilde{S} := \sum_{j=1}^k |X_j|^{1/2}$.  By H\"{o}lder's Inequality, for any $x \in \rd$ and $p >q>0$,
\begin{equation*}
    \norm{x}_p \leq \norm{x}_q \leq d^{1/q-1/p} \norm{x}_p.
\end{equation*}
Consequently, for $x = (X_1, \ldots, X_k) \in \rk$, $p=1$ and $q=1/2$ we have that $S \leq \tilde{S}^2 \leq k \cdot S$, hence for any $t>0$,
\begin{equation} \label{normineq}
    \prob[S \leq t] \leq \prob[\tilde{S} \leq \sqrt{tk}].
\end{equation}  

We use the bound on the moment-generating function, to prove a Chernoff-type concentration bound for $\tilde{S}$, which by Eq.~\eqref{normineq} translates into a concentration bound for $S$.

\begin{lemma} \label{conc-bound}
For every $D>1$,
\begin{equation*}
    \prob \left [ \tilde{S} \leq \frac{\EE [\tilde{S}]}{D}\right] \leq \left( \frac{10}{D}\right)^k.
\end{equation*}
\end{lemma}
\begin{proof}
Since $X_j$'s are independent, $\EE[\tilde{S}] = \sqrt{2}k$. Then, by Lemma \ref{lemma:mgf} and Markov's inequality, for any $\beta>1$, it follows that
\begin{align*}
    \prob \left [ \tilde{S} \leq \frac{\EE [\tilde{S}]}{D}\right] 
    &= \prob \left[  \exp (-\beta \tilde{S}) \geq  \exp \left(-\beta\cdot  \frac{\EE [\tilde{S}]}{D}\right)\right] \\
    &\leq \frac{\EE [\exp (-\beta \tilde{S})]}{\exp (-\beta \EE [\tilde{S}]/D)} \\
    &= \frac{\EE [\exp (-\beta |X_j|^{1/2})]^k}{\exp (-\beta \sqrt{2}k/D)} \\
    &\leq \left(\frac{2}{\beta} \right)^k \cdot \ex^{\sqrt{2} \beta k/D}.
\end{align*}

Setting $\beta=D$ completes the proof. 
\end{proof}

\section{Net-based dimension reduction}\label{Sdimreduc}
In this section we describe the dimension reduction mapping for $\ell_1$ via $r$-nets. Let $P\subset \ell_1^d$ be a set of $n$ points with doubling constant $\lambda_P$. For some point $x \in \rd$ and $r>0$, we denote by $B_1(x,r)$ the $\ell_1$-ball of radius $r$ around $x$. The embedding is non-linear and is carried out in two steps.

First, we compute a $c$-approximate $(\eps/c)$-net $\net$ of $P$ with the algorithm of Theorem \ref{theorem:apprnets}. Moreover, the algorithm assigns each point of $P$ to the point of $\net$ which covered it. Let $g: P \rightarrow \net$ be this assignment. In the second step, for every $s \in \mathcal{N}$ and any query point $q \in \ell_1^d$, we apply the linear map of Theorem \ref{theorem:indyk}. That is,  $f(s) = As/T$, where $A$ is a $k{\times}d$ matrix with each element being an i.i.d. Cauchy random variable. Recall that value $T = \Theta(k \log {(k/ \eps)})$. By the $1$-stability property of the Cauchy distribution, $f(s)$ is distributed as $\norm{s}_1 \cdot (Y_1, \ldots, Y_k)$, where each $Y_j$ is i.i.d. and $Y_j \sim \cauchy$. Hence, $\norm{f(s)}_1 = \norm{s}_1 \cdot S$ where $S := \sum_j |Y_j|$.

We define the embedding to be $h = f \circ g$. We apply $h$ to every point in $P$, and $f$ to any query point $q$. It is clear from the properties of the net that $g$ incurs an additive error of $\pm \eps$ on the distance between $q$ and any point in $P$, so it is sufficient to consider the distortion of $f$. 

Our analysis consists of studying separately the following disjoint subsets of $\net$: Points that lie at distance at most $D_0$ from the query and points that lie at distance at least $D_0$, for some $D_0>1$ chosen appropriately. For the former set, we directly apply Theorem \ref{theorem:indyk}, as it has bounded diameter.

The next lemma guarantees the low distortion for points of the latter set, namely those that are sufficiently far from the query. We consider the sum of the square roots of each $|Y_j|$, i.e., $\tilde{S} = \sum_j |Y_j|^{1/2}$, in order to employ the tools of Section~\ref{Sconc-bound}.

\begin{lemma} \label{lemma:farpoints}
Fix a query point $q \in \ell_1^d$. For any $\eps \leq 1/2$, $c \geq 1$, $\delta \in (0,1)$, there exists $D_0=O(\log (k/\eps))$ such that  for $k = \Theta \left( \log^{2}{\lambda_P} \cdot \log ( c / \eps) + \log (1/ \delta) \right)$, with probability at least $1 - \delta$,
\begin{equation*}
    \forall s \in \net :~ \norm{s-q}_1 \geq D_0 \implies \norm{f(s)-f(q))}_1 \geq 4.
\end{equation*}
\end{lemma}

\begin{proof}
Assume wlog that the query point is the origin $(0,\dots,0)$. For some $D_0>1$, we define the following subsets of $\net$:
\begin{equation*}
    N_i := \{ s \in \mathcal{N} \ | \ D_i \leq \norm{s}_1 < D_{i+1} \}, \ D_i = 2^{2i} D_0, \ i = 0, 1, 2, \ldots
\end{equation*}
By the definition of doubling constant and the fact that two points of $\net$ lie at distance at least $\eps$, 
\[ |N_i| \leq \lambda_P^{\lceil \log \left(4cD_{i+1} / \eps \right) \rceil} \leq \lambda_P^{4 \log \left(cD_{i+1} / \eps \right)}. \] 
Therefore, by the union bound, and Eq.~\eqref{normineq}:
\begin{align*}
    \prob \left[ \exists i \exists s \in N_i : \norm{f(s)}_{1} \leq \frac{4\norm{s}_{1}}{D_i}  \right] &= \prob \left[ \exists i \exists s \in N_i : S \leq \frac{4T}{D_i}  \right]\\
    &\leq \sum_{i=0}^{\infty} |N_i| \prob \left[\tilde{S} \leq \frac{\sqrt{4kT}}{\sqrt{D_i}} \right] \\
    &= \sum_{i=0}^{\infty} |N_i| \prob \left[\tilde{S} \leq \EE[\tilde{S}] \cdot \sqrt{\frac{2T}{k 2^{2i} D_0}} \right]. 
\end{align*}
By Lemma~\ref{conc-bound}, for $D_0 {=}\lceil 800T/k \rceil {=} \Theta(\log (k/\eps))$ and $k > 4 {\cdot} \log{\lambda_P} {\cdot} \log(c D_0/ \eps)+2 \log (2 \lambda_P/\delta) $:
\begin{align*}
    \sum_{i=0}^{\infty} |N_i| \prob \left[\tilde{S} \leq \frac{\EE[\tilde{S}]}{10 \cdot 2^{i+1}} \right] &\leq \sum_{i=0}^{\infty} \lambda_P^{4 \log{(cD_{i+1}/ \eps)}} \left( \frac{1}{2^{i+1}}\right)^k \\
    &= \sum_{i=0}^{\infty} \frac{ 2^{\log(\lambda_P)({4 \log{(cD_0/ \eps)}+2i+2})}}{ 2^{k(i+1)}} \\
    &\leq \sum_{i=0}^{\infty} \frac{ 2^{\log(\lambda_P) \cdot 4 \log{(cD_0/ \eps)}} \cdot 2^{2\log(\lambda_P) (i+1)}}{ 2^{(4 \cdot \log{\lambda_P} \cdot \log(c D_0/ \eps))(i+1)} \cdot 2^{2 \log (2 \lambda_P/\delta))(i+1)}} \\
    &\leq \sum_{i=0}^{\infty} { 2^{-2 \log (2/\delta))(i+1)} } \\
    &= \sum_{i=0}^{\infty} \left( \frac{\delta^2}{4} \right)^{i} - 1 \\
    &= \frac{\delta^2}{4 - \delta^2} \\
    &\leq \delta.
\end{align*}

Finally, for some large enough constant $C$, we demand that
\begin{equation*}
    k > C \left( \log{\lambda_P} \cdot \log ( c \log k/ \eps) + \log (1/ \delta) \right) > 4 \cdot \log{\lambda_P} \cdot \log(c D_0/ \eps)+2 \log (2 \lambda_P/\delta)
\end{equation*}
which is satisfied for $k = \Theta \left( \log^{2}{\lambda_P} \cdot \log ( c / \eps) + \log (1/ \delta) \right)$.
\end{proof}

\begin{theorem} \label{theorem:main}
Let $P \subset \ell_1^d$ such that $|P|=n$. For any $\eps \in (0, 1/2)$ and $c \geq 1$, there is a non-linear randomized embedding $h = f \circ g: \ell_1^d \rightarrow \ell_1^k$, where $k = \left( \log{\lambda_P} \cdot \log(c / \eps) \right)^{\Theta(1/\eps)} / \zeta(\eps) $, for a function $\zeta(\eps)>0$ depending only on $\eps$, such that, for any  $q \in \ell_1^d$ , if there exists $p^{*}\in P$ such that $\|p^{*}-q\|_1 \leq 1$, then, with probability $\Omega(\eps)$:
\begin{align*}
    &\norm{h(p^{*}) - f(q)}_1 \leq 1 + 3\eps, \\
    &\forall p\in P:~\|p-q\|_1 > 1+9\eps \implies \norm{h(p) - f(q)}_1 > 1 + 3\eps.
\end{align*}
Set $P$ can be embedded in time $\tilde{O}(d n^{1+1/\Omega(c)})$, and any query $q\in \ell_1^d$ can be embedded in time $O(dk)$.
\end{theorem}

\begin{proof}
Let $f, g$ be the mappings defined in the beginning of the section and $D_0 = \Theta(\log (k/\eps))$. Assume wlog for simplicity that $q=0^d$. Then, by Lemma \ref{lemma:farpoints} for $k = \Theta \left( \log^{2}{\lambda_P} \cdot \log ( c / \eps) \right)$, with probability at least $1- \eps/5$, we have:
\begin{equation*}
    \forall p \in P :~ \norm{p-q}_1 \geq D_0 + \eps \implies \norm{h(p) - f(q)}_1 \geq 4.
\end{equation*}
By Theorem~\ref{theorem:indyk}, for $\gamma = \eps/10$ and $\delta = \eps/ (5\lambda_P^{8\log{(c D_0 / \eps)}})$, with probability at least $1-\eps/5$, we get:
\[
\forall p\in P:~\|p-q\|_1 \in (1+9\eps, D_0+ \eps)  \implies \norm{h(p) - f(q)}_1 > (1 + 8\eps)(1 - \eps) \geq 1 + 3\eps.
\]
Moreover,
\[
\prob \big[\norm{h(p^*) - f(q)}_1 \leq 1 + 3\eps  \big] \geq 1 -\frac{1+ \eps/10}{1+ \varepsilon} \geq 1 - (1 - \eps/2).
\]
Then, the target dimension needs to satisfy the following inequality:
\begin{align*}
    k \geq \frac{\big( \ln{ ( 5\lambda_P^{8\log{(c D_0 / \eps)}} / \eps ) } \big) ^{2/\eps}}{ \zeta(\eps)} = \frac{\big( \Theta ( \log \log k  \cdot \log{\lambda_P} + \log{\lambda_P} \cdot \ln(c / \eps) ) \big) ^{2/\eps}}{ \zeta(\eps)}.
\end{align*}
Hence, for $k = \left( \log{\lambda_P} \cdot \log(c / \eps) \right)^{\Theta(1/\eps)} / \zeta(\eps) $, we achieve a total probability of success in $\Omega(\eps)$, which completes the proof.
\end{proof}

\section{Dimension reduction based on randomly shifted grids}\label{Sgrids}
In this section, we explore some properties of randomly shifted grids, and we present a simplified embedding which consists of a first step of snapping points to a grid, and a second step of randomly projecting grid points. 

Let $w >0$ and $t$ be chosen uniformly at random from the interval $[0,w]$. The function 
\[h_{w, t}(x)=\left\lfloor \frac{x-t}{w} \right\rfloor\]
induces a random partition of the real line into segments of length $w$. Hence, the function 
\[g_{w}(x)=(h_{w,t_1}(x_1),...,h_{w,t_d}(x_d)),\]
for $t_1,\ldots,t_d$ independent uniform random variables in the interval $[0,w]$, 
induces a randomly shifted grid in $\rd$. For a set $X \subseteq \rd$, we denote by $g_{w}(X)$, the image of $X$ on the randomly shifted grid points defined by $g_{w}$. For some $x \in \rd$ and $r>0$, the number of grid cells of $g_{w}(\ell_1^d)$ that $B_1(x,r)$ intersects per axis is independent, and in expectation is $1{+}2r/w$. Then, the expected total number of grid cells that $B_1(x,r)$ intersects is at most $(1{+}2r/ w)^d$.

Now let $P \subset \ell_1^d$ be a set of $n$ points with doubling constant $\lambda_P$ and $q \in \ell_1^d$ a query point. For $w = \eps/d$, the $\ell_1$-diameter of each cell is $\eps$ and therefore $g_{w}(P)$ is an $\eps$-covering set of $P$.

\begin{lemma} \label{lemma:grid1}
Let $\mathcal{R}>1$ and $P' := B_1(q,\mathcal{R}) \cap P$. Then, for $w=\eps/d$
\[
\EE \big[ |g_{w}(P')| \big] \leq 8 \lambda_P^{2\log (d \mathcal{R}/\eps)}.
\] 
\end{lemma}
\begin{proof}
By the doubling constant definition, there exists a set of balls of radius $\eps/d^2$ centered at points in $P'$, of cardinality at most $\lambda_P^{2\log (d \mathcal{R}/\eps)}$ which covers $P'$. For each ball of radius $\eps/d^2$, the expected number of intersecting grid cells is $(1{+}2/d)^d \leq \ex^2$. The lemma follows by linearity of expectation.
\end{proof}

The next lemma shows that, with constant probability, the growth on the number of representatives, as we move away from $q$, is bounded. 

\begin{lemma}\label{lemma:grid2}
Let $\{D_i\}_{i \in \mathbb{N}}$ be a sequence of radii such that, for any $i$, $D_{i+1}=4 D_i$. Let $A_i$ be the points of $g_w(P)$ within distance $D_{i+1} = 2^{2(i+1)} D_0$ from $q$. Then, with probability at least $1/3$,
\[
\forall i\in \{-1,0,\ldots\}: |A_i| \leq 4^{i+3} \lambda_P^{2\log (d D_{i+1}/\eps)}.
\]
\end{lemma}

\begin{proof}
By Lemma \ref{lemma:grid1}, $\EE[ |A_i|] \leq 8 \lambda_P^{2\log (d D_{i+1}/\eps)}$ for every $i\in \{-1,0,\ldots\}$. Then, a union bound followed by Markov's inequality yields
\begin{equation*}
    \prob \big[ \exists i \in \{0,1,\ldots\} : |A_i| \geq {4^{i+1}} \EE[ |A_i|] \big] \leq  1/3.
\end{equation*}
In addition,
\begin{equation*}
    \prob \big[  |A_{-1}| \geq 4 \EE[ |A_i|] \big] \leq  1/4. \qedhere
\end{equation*}
\end{proof}

\begin{theorem}
\label{theorem:second}
Let $P\subset \ell_1^d$ such that $|P|=n$. 
For any $\eps \in (0,1/2)$, there is a non-linear randomized embedding $h': \ell_1^d \rightarrow \ell_1^k$, where $k = \left( \log{\lambda_P} \cdot \log(d / \eps) \right)^{\Theta(1/\eps)} / \zeta(\eps)$, for a function $\zeta(\eps)>0$ depending only on $\eps$, such that for any  $q \in \ell_1^d$ , if there exists $p^{*}\in P$ such that $\|p^{*}-q\|_1 \leq 1$, then with probability $\Omega(\eps)$,
\begin{align*}
    &\norm{h'(p^{*}) - f(q)}_1 \leq 1 + 3\eps, \\
    &\forall p\in P:~\|p-q\|_1 > 1+9\eps \implies \norm{h'(p) - f(q)}_1 > 1 + 3\eps.
\end{align*}
Any point can be embedded in time $O(dk)$.
\end{theorem}

\begin{proof}
We follow the same reasoning as in the proof of Theorem \ref{theorem:main}. The embedding is $h' = f \circ g_{\eps/d}$, where $f$ is the randomized linear map defined in Section~\ref{Sdimreduc}. As before, we apply $h'$ to every point in $P$, and only $f$ to queries. The randomly shifted grid incurs an additive error of $\eps$ in the distances between $q$ and $P$.

Assume wlog that $q = 0^d$ and let $A_i$ be the points of $g_{\eps /d}(P)$ within distance $D_{i+1} = 2^{2(i+1)} D_0$ from $q$. Hence, by Lemma \ref{lemma:grid2},
\begin{align*}
    \prob \left[ \exists i \exists s \in A_i : \norm{f(s)}_{1} \leq \frac{4 \norm{s}_{1}}{D_i}  \right] &\leq \sum_{i=0}^{\infty} |A_i| \prob \left[ S \leq \frac{4T}{{D_i}} \right]\\
    &\leq \sum_{i=0}^{\infty} 4^{i+3} \lambda_P ^ {2\log (d D_{i+1}/\eps )} \prob \left[\tilde{S} \leq \frac{\sqrt{4kT}}{\sqrt{D_i}} \right].
\end{align*}
As in Lemma \ref{lemma:farpoints}, for $D_0=\left\lceil{800 T}/{k} \right\rceil = \Theta(\log{(k/ \eps)})$, $k\geq 20 \log{\lambda_P} \cdot \log \left(\frac{dD_0}{\eps\delta}\right)$ and $\delta=\eps/5$,
\[
\sum_{i=0}^{\infty} 4^{i+3} \lambda_P ^ {2\log (d D_{i+1}/\eps )} \prob \left[\tilde{S} \leq \frac{\sqrt{4kT}}{\sqrt{D_i}} \right] \leq \sum_{i=0}^{\infty} \frac{2^{2i+6+2 \log \lambda_P [\log (d D_{0}/\eps) +2(i+1)]}}{2^{k(i+1)}} \leq \eps/5.
\]
Hence, for $k =\Omega \left((\log^2{\lambda_P} \cdot \log (d / \eps) \right)$, with probability at least $1- \eps/5$, we have:
\begin{equation*}
    \forall p \in P :~ \norm{p-q}_1 \geq D_0 + \eps \implies \norm{h'(p) - f(q)}_1 \geq 4.
\end{equation*}

Now, we are able to use Theorem \ref{theorem:indyk} for points which are at distance at most $D_0 + \eps$ from $q$, and the near neighbor. By Lemma \ref{lemma:grid2}, with constant probability, the number of grid points at distance $\leq D_0 + \eps$, is at most $32\cdot \lambda_P^{4 \log (dD_0/\eps)}$. Hence, by Theorem \ref{theorem:indyk}, for $\gamma = \eps/10$ and $\delta = \eps/ (160\lambda_P^{4\log{(d D_0 / \eps)}})$, with probability at least $1-\eps/5$, it holds:
\[
\forall p\in P:~\|p-q\|_1 \in (1+9\eps, D_0+ \eps)  \implies \norm{h'(p) - f(q)}_1 > 1 + 3\eps.
\]
Moreover, with probability at least $\eps/2$, we obtain:
\[
\norm{h'(p^*) - f(q)}_1 \leq 1 + 3\eps.
\]
As in Theorem \ref{theorem:main}, the target dimension needs to satisfy the following:
\[
k \geq \frac{\big( \ln{ ( 160\lambda_P^{4\log{(d D_0 / \eps)}} / \eps ) } \big) ^{2/\eps}}{ \zeta(\eps)}.
\]
Hence, for $k = \left( \log{\lambda_P} \cdot \log(d / \eps) \right)^{\Theta(1/\eps)} / \zeta(\eps) $ we achieve total probability of success $\Omega(\eps)$.
\end{proof}

\section{Conclusion}

We have filled in a gap in the spectrum of randomized embeddings with bounded distortion only for distances between the query and a pointset: such embeddings existed for $\ell_2$ and $\ell_1$ and for doubling subsets of $\ell_2$. Here we settle the case of doubling subsets of $\ell_1$ with a \textit{near} neighbor-preserving embedding.  In the meantime, we obtain concentration bounds on sums of independent Cauchy variables. 
Our algorithms are quite simple, therefore they should also be of practical interest.

We rely on approximate $r$-nets or randomly shifted grids. For the former,
Theorem~\ref{theorem:main} provides with a trade-off between the preprocessing time required and the target dimension. On the other hand, Theorem~\ref{theorem:second} has the advantage of fast preprocessing: any point is embedded in $O(dk)$ time, and the embedding is oblivious to the pointset. In regards to the near-linear preprocessing time, the two results are comparable, since the dimension in Theorem~\ref{theorem:second} can be substituted by the target dimension of Theorem~\ref{theorem:indyk}.  

Notice that any potential improvements to Theorem~\ref{theorem:indyk} should lead to improvements to Theorems~\ref{theorem:main} and~\ref{theorem:second}. The target dimension in these theorems follows from a direct application of Theorem~\ref{theorem:indyk} to the representative data points which lie inside a bounding ball centered at the query. 

\bibliography{main}
\end{document}